\documentclass[10pt, conference, letterpaper]{IEEEtran}
\usepackage{graphicx}
\usepackage{caption}	
\usepackage{subcaption}
\usepackage{amsmath}
\usepackage{amsthm}
\usepackage{soul,color}
\usepackage{todonotes}
\usepackage{url}
\usepackage{mathtools}
\usepackage{calc}
\usepackage{xparse}
\usepackage{algorithm}
\usepackage{algpseudocode}
\usepackage{pifont}
\usepackage{makeidx}
\usepackage{relsize}
\usepackage{mathrsfs}
\usepackage[utf8]{inputenc}
\usepackage{multicol}
 \usepackage{multirow}
\usepackage{stfloats}
\usepackage[english]{babel}
\usepackage{mathtools, cuted}
 \usepackage{cite}
 \usepackage{epstopdf}
\DeclareMathAlphabet{\mathpzc}{OT1}{pzc}{m}{it}
\newtheorem{theorem}{Theorem}

\begin{document}
\title{Dynamic Base Station Repositioning to Improve Spectral Efficiency of Drone Small Cells}

\author{\IEEEauthorblockN{Azade Fotouhi\IEEEauthorrefmark{1}\IEEEauthorrefmark{2},
Ming Ding\IEEEauthorrefmark{2} and
Mahbub Hassan\IEEEauthorrefmark{1}\IEEEauthorrefmark{2}}
\IEEEauthorblockA{\IEEEauthorrefmark{1}School of Computer Science and Engineering\\
University of New South Wales (UNSW), Sydney, Australia\\
Email: \{a.fotouhi, mahbub.hassan\}@unsw.edu.au}
\IEEEauthorblockA{\IEEEauthorrefmark{2}Data61, CSIRO, Australia\\
Email: \{azade.fotouhi, ming.ding, mahbub.hassan\}@data61.csiro.au}
}
\maketitle

\begin{abstract}
With recent advancements in drone technology, 
researchers are now considering the possibility of deploying small cells served by base stations mounted on flying drones. 
A major advantage of such drone small cells is that the operators can quickly provide cellular services in areas of urgent demand without having to pre-install any infrastructure. 
Since the base station is attached to the drone, 
technically it is feasible for the base station to dynamic reposition itself in response to the changing locations of users for reducing the communication distance, 
decreasing the probability of signal blocking, 
and ultimately increasing the spectral efficiency. 
In this paper, 
we first propose distributed algorithms for autonomous control of drone movements, 
and then model and analyse the spectral efficiency performance of a drone small cell to shed new light on the fundamental benefits of dynamic repositioning.
We show that, 
with dynamic repositioning, 
the spectral efficiency of drone small cells can be increased by nearly 100\% for realistic drone speed, height, and user traffic model and without incurring any major increase in drone energy consumption.
\end{abstract}
\IEEEpeerreviewmaketitle

\section{Introduction} \label{sec:intro}
A drone is an unmanned aerial vehicle designed to be flown either through remote control or autonomously using embedded software and sensors. 
Historically, 
drones had been used mainly in military for reconnaissance purposes, 
but with recent developments in light-weight battery-powered drones, 
many civilian applications are emerging. 
Use of drones to deploy small cells in areas of urgent needs is one of the most interesting applications currently being studied by many researchers~\cite{Namuduri,6399121,7122576,7417609,7451189,7122575}. 
The greatest advantage of this approach is that drones can be equipped with small cell base station (BS) module and sent to a specific target location immediately to establish emergency communication links without having to deploy any infrastructure.

In this paper, 
we go beyond the basic advantage of quick deployment and seek further benefits that these emerging agile drones could bring to cellular communication networks. 
Since the base station is attached to the drone, 
technically it is feasible for the base station to dynamic reposition itself in response to the changing users locations for reducing the communication distance, 
decreasing the probability of signal blocking, 
and ultimately increasing the spectral efficiency. 
While \textit{dynamic repositioning} of BSs is a conceptually simple idea, 
its realisation would require development of distributed algorithms for autonomous control of continuous drone movement that would optimise the spectral gain. 
It is also not immediately clear whether the benefits of dynamic repositioning would be significant for practical use, 
given that drone flying speeds must be restricted to conserve battery life and to ensure safety operations.
Moreover, 
the impact of user traffic model on the spectral efficiency gain will also have to be carefully analysed before any conclusions can be made about the performance improvement of dynamic repositioning the drone BSs.

Recent studies~\cite{7417609,7451189, 7122575} on drone small cells mainly focused on finding a stationarily optimal location in the air for the drone to \textit{hover}, 
while serving the target area with a given population and traffic demand. 
To the best of our knowledge, 
the concept of dynamic repositioning drone BSs has not been adequately explored in the literature to quantify its benefits in terms of the spectral efficiency gain. 
In this paper, 
we model and analyse the spectral efficiency of a drone small cell with or without dynamic repositioning and shed new light on the fundamental benefits of dynamic repositioning. 
Our analysis shows that, 
with dynamic repositioning, 
the spectral efficiency of drone small cells can be increased by nearly 100\% for realistic drone speed, height, and users traffic model and without incurring any major increase in drone energy consumption.

The novelty and contributions of this paper are summarised as follows:
\begin{itemize}
\item We propose an analytical model for a single drone small cell serving a single user, 
and derive the spectral efficiency gain for dynamic repositioning. 
Our model sheds new light on the fundamental benefits of dynamic repositioning for drone small cells as a function of drone flying speed, height and user traffic parameters. 
Our analytical results confirm that dynamic repositioning can increase spectral efficiency of drone cells \textit{significantly}.

\item We propose two distributed algorithms for autonomous dynamic repositioning of the drone with an objective to optimize the spectral efficiency. 
    The first algorithm, 
    which is simple to implement, 
    enables the drone to make movement decisions using only local position information of the active user that it is currently serving. 
    The second algorithm requires information about user locations in neighbor cells to minimize \textit{interference leakage} to the surrounding and maximize the overall spectral efficiency of the entire network area. 
    Both algorithms are evaluated in realistic multi-cell environments.

\item We simulate the proposed dynamic positioning algorithms using realistic drone and traffic parameters with multiple cells and multiple mobile users in each cell. 
    As for the medium access scheme to serve multiple users in each cell, 
    we consider two widely used methods of resource allocation, 
    i.e., 
    time division multiple access (TDMA) and frequency division multiple access (FDMA). 
    Our simulations show that, 
    even with the simplistic dynamic repositioning algorithm that requires the local position information only, 
    the spectral efficiency of drone small cells can be increased by nearly 100\% without any negative effect on drone battery life.
    This result is achievable with both TDMA and FDMA,
    making dynamic repositioning a practical solution for a wide range of existing platforms. 
    A further 8\% increase in the spectral efficiency can be achieved with the more advanced dynamic repositioning algorithm that requires the neighboring position information to minimize interference leakage to adjacent cells.

\end{itemize}

The remainder of the paper is structured as follows. 
Related work is reviewed in Section~\ref{sec:relatedwork}. 
Section~\ref{sec:systemmodel} introduces the system model. 
In Section~\ref{sec:analymodel}, 
we present the analytical model and theoretical results for a single drone cell. 
Our proposed algorithms and simulation results for a network with multiple drone cells are presented in Sections~\ref{sec:proposedmethod} and~\ref{sec:simulationresults},
respectively.
Finally, we conclude our paper in Section~\ref{sec:conclusion}.

\section{Related Work} \label{sec:relatedwork}

Drones have been considered both in the context of data gathering in wireless sensor networks~\cite{6848000,7192644,7090210}, 
and more recently in the context of delivering data to mobile users in cellular networks. 
Since the focus of this paper is on cellular networks, 
we only review the drone-related research relevant to cellular networks.

Because of the flexibility and agility of drones, 
deploying the drone base stations (BSs) in optimal locations to maximize various network metrics are investigated in the literature.
Al-Hourani et al.~\cite{al2014optimal} provided an analytical model to find an optimal altitude for one UAV that provides the maximum coverage of the area. 
A service threshold in terms of maximum allowable path loss is defined in this model. 
Another recent study by Mozaffari et al.~\cite{7510870} studied the problem of finding the optimal cell boundaries and deployment location for multiple non-interfering UAVs. 
The objective of that study was to minimize the total transmission power of the UAVs. 

UAVs were also expected to establish emergency communication links during disaster situation and thus improve public safety in~\cite{7122576}. 
That work showed that by optimal placement of UAVs the system throughput can be improved significantly. 
Brute force search was used to find the optimal location of UAVs in the target area. 

Finding the 3D optimal location for deploying a drone cell was studied in~\cite{yaliniz2016}. 
When some users with QoS requirements are distributed in an area, 
a 3D location could be found for deploying a drone cell to provide services for the maximum number of users satisfying their SNR (Signal to Noise Ratio) constraints. 

Multiple interfering UAVs bring more challenges such as the distance between UAVs. 
The problem of the optimal deployment of two interfering drone small cells is investigated in~\cite{7417609}. 
Rohde et al.~\cite{6399121} addressed the problem of cell outage or cell overload using UAVs to temporarily offload the traffic to neighbor cells in 4G networks. 
In that work, 
a central planning model for the placement of UAV relays was discussed, 
the feasibility of the solution was also proved using an analytical model. 

Moreover, a propagation model for air-to-ground communication is studied in~\cite{7037248} and~\cite{al2014optimal}. 
Those studies derived a closed-form expression based on elevation angle to present the probability of having LoS (Line of Sight) connection for low altitude platforms. 

There are a few other studies in the area of optimizing the deployment of the drone BSs. 
However, 
they have not considered the mobility of drones after the deployment, 
to see how dynamic repositioning of the drones can improve the network performance.

Note that in our previous work~\cite{7848883}, 
we focused on the spectral efficiency gain of dynamically repositioning a single drone. 
Here, 
we consider multiple interfering drone BSs, 
devise various heuristic algorithms, 
and validate them through extensive simulation.
In addition, 
we consider practical commercial drone BSs in this work, 
which mandates that drones should not fly too high or too fast. 
In contrast, 
the existing studies~\cite{7192644,7417609} usually considered the application of drones in wireless sensor networks or city-wide broadcasting services, 
where drones can fly at a height of several hundreds of meters and at a speed of several tens of kilometers per hour. 
More specifically, 
we propose that drone BSs should not fly too high in cellular networks because,
\begin{itemize}
  \item (a) Drones flying at a high altitude, e.g., larger than 10$\sim$20 meters, 
  will significantly increase the inter-cell interference as shown in the 4G/5G networks~\cite{Ding2016GC_ASECrash}.
  \item (b) Drones flying at a higher altitude will cause much more energy consumption than that at a lower altitude~\cite{ZORBAS201380,Rubio2016}.
  \item (c) Drones flying at a high altitude are susceptible to wind, 
  thus they will have to be made heavy and large, 
  which are not safe to use in practice in case of crash.
\end{itemize}

Besides, we propose that drone BSs should not fly too fast in cellular networks because,
\begin{itemize}
  \item (a) Drones flying at a higher speed will cause much more energy consumption than that at a lower speed \cite{7101619}.
  \item (b) Drones flying at a high speed will cause tremendous damage in case of collision.
\end{itemize}

As practical assumptions, 
we propose that the flying height of drone BSs should be around 10 meters as widely used by small cell BSs in the 3GPP LTE networks \cite{TR36.828,3gpp36814}, 
and the flying speed of drone BSs should be around 10m/s, 
which does not cause extra energy consumption compared with drone hovering as shown in the measurement study for the state-of-the-art drones~\cite{7101619}.

With the above practical assumptions, 
it is particularly interesting to investigate \emph{the spectral efficiency gain of dynamic repositioning of drone BSs in the cellular networks}, 
\emph{if there exists some gain at all}. 
We will answer this fundamental question in the following sections.

\section{System Model} \label{sec:systemmodel}

In this section, 
we describe the system model for the proposed dynamic repositioning of drone BSs. 
When a drone is moving across its cell,
the distance between the drone and its users will change and in turn affect the strength of the received signals.
As a result,
the drone needs to update its moving direction constantly according to the position of active users who have data requests.
The direction should be selected in a way to maximize the spectral efficiency of the network.

\subsection{Network Scenario}\label{sec:overall}

Let's consider a network area consisting of $N$ cells,
with each cell having a pre-determined shape,
e.g., square (with a edge length of $l$ (in meter)), 
or disk (with a radius of $R$ (in meter)).
In each considered cell,
there is a drone that can either move with a constant speed of $v$ (in m/s) or hover at a low height of $h$ (in meter).
There are $U$ mobile users in each cell. 
Initially all drones are located above the center of their cells.

Each drone, 
which may be connected to a nearby macro cell tower with a wireless back-haul link, 
is responsible to provide wireless communication services for the users in its cell.

The \textit{ground distance} or the two-dimensional (2D) distance between user $u$ and drone $n$ is defined by the distance between the user and the projection of the drone location onto the ground,
denoted by $r_{u,n}$.
The \textit{Euclidean distance} or the three-dimensional (3D) distance between user $u$ and drone $n$ is presented by $d_{u,n} = \sqrt{r_{u,n}^2 +h^2}$.

We further assume that each drone is transmitting data to users with a fixed transmission power of $P_{tx}$ (in watt),
$B$ (in Hz) is the total bandwidth,
and $f$ (in Hz) is the central carrier frequency.
As we are considering a multi-user system,
the amount of allocated bandwidth to an active user $u$ in a cell is denoted by $b_u$ ($1 \leq u \leq U$, $0 \leq b_u \leq B$).

\subsection{Traffic Model}\label{sec:traffic_model}

The traffic model for each user follows the recommended traffic model by 3GPP \cite{3gpp36814}.
More specifically,
we define two times, i.e., reading time and transmission time,
\begin{itemize}
  \item The reading time is defined as the time interval between the end of the download of the previous data package and the user request for the next one.
      In this paper,
      the reading time of each data package is modeled as an exponential distribution with a mean of $\lambda$ (in sec).
  \item The transmission time is denoted by $\tau$ and defined as the time interval between the request time of a data package and the end of its download.
\end{itemize}

Based on the above definition,
during a transmission time $\tau$,
the associated user is called an \textit{active} user.
The set of all active users in a cell $n$ at a specific time $t$ is denoted by $\mathpzc{Q}_{n}(t)$ ($1 \leq n\leq N$).
At a user's data requesting time,
the user will request a data package with a fixed size $s$ (in MByte).

\subsection{Channel Model} \label{sec:channel}

In this paper,
we consider a practical path loss model incorporating both LoS (Line of Sight) transmissions and NLoS (Non Line of Sight) transmissions.
More specifically,
the path loss function is formulated according to a probabilistic LoS model~\cite{al2014optimal},
in which the probability of having a LoS connection between a drone and its user depends on the elevation angle of the transmission link.
According to~\cite{al2014optimal},
the LoS probability function is expressed as
\begin{equation}
P^{LoS}(h,r_{u,n}) = \frac{1}{1+\alpha exp(-\beta[\theta -\alpha])},
\label{eq:plos}
\end{equation}
where $\alpha$ and $\beta$ are environment-dependent constants,
$\theta$ equals to $arctan(h/r_{u,n})$ in degree,
and $h$ denotes the drone height as discussed before.
As a result of (\ref{eq:plos}),
the probability of having a NLoS connection can be written as
\begin{equation}
P^{NLoS}(h,r_{u,n}) = 1 - P^{LoS}(h,r_{u,n}).
\label{eq:pnlos}
\end{equation}

From (\ref{eq:plos}) and (\ref{eq:pnlos}),
the path loss in dB can be modeled as
\begin{equation}
\eta_{path}(h,r_{u,n}) = A_{path} + 10\gamma_{path}\log_{10}(d_{u,n}),
\label{eq:pathloss}
\end{equation}
where the string variable \textit{"path"} takes the value of “LoS” and “NLoS” for the LoS and the NLoS cases, respectively.
In addition,
$A_{path}$ is the path loss at the reference distance (1 meter) and $\gamma_{path}$ is the path loss exponent,
both obtainable from field tests \cite{TR36.828}.

\subsection{MAC Layer} \label{sec:resourceallocation}

Regarding the multiple access scheme to allow multiple users to be served by drone BSs, 
we study two resource allocation methods, 
namely FDMA and TDMA, 
which will be explained in the following subsections.

\subsubsection{FDMA}\label{sec:FDMA}

In the FDMA method,
each drone simultaneously serves all the active users by \emph{equally} dividing the whole available bandwidth among them.
When a transmission for a request finishes,
the amount of released bandwidth is again shared \emph{equally} among the unfinished requests.
Subsequently,
in this method all active users in a cell will receive data from their serving drone BS. 
It is important to note that the orthogonal FDMA (OFDMA) technique adopted in the 4G LTE networks divides the frequency channel into sub-carriers,
and thus it is merely a special case of FDMA.

Also note that the equal bandwidth division ensures fairness among the active users.
To achieve different levels of trade-off between performance and fairness, 
other unequal resource allocation methods can be studied in our work.
Since it is straightforward to do so and does not affect our framework,
we omit such discussion for brevity.

\subsubsection{TDMA} \label{sec:tdma}

In the TDMA method,
instead of allocating the bandwidth equally to all the active users,
the whole bandwidth is devoted to only one user during each time slot.
Among all the active users,
the one that has the maximum received signal strength is selected.

\subsection{Performance Metrics} \label{sec:metrics}

The Signal to Noise Ratio (SNR) of user $u$ connected to drone $n$ is defined as
\begin{align}
 \begin{split}
SNR_u(h,r_{u,n})=\frac{S^{path}(h,r_{u,n})}{N_u},
\end{split}
 \label{eq:snr}
\end{align}
where $N_u$ (in watt) represents the total noise power including the thermal noise power and the user equipment noise figure,
which is given by \cite{thermalnoise} 
\begin{equation}
N_u = 10^{ \frac{-174+\delta_{ue}}{10}}\times{b_u}\times10^{-3},
\label{eq:noise}
\end{equation}
where $\delta_{ue}$ is the user equipment noise figure (in dB).

In (\ref{eq:snr}),
$S^{path}(h,r_{u,n})$ (in watt) indicates the received power,
which can be obtained by
\begin{align}
 \begin{split}
S^{path}(h,r_{u,n})&=\frac{b_u}{B} P_{tx}10^{\frac{-\eta_{path}(h,r_{u,n})}{10}}\\
							&=\frac{b_u}{B} P_{tx}A'_{path}(d_{u,n})^{-\gamma_{path}},
\end{split}
\label{eq:rcvpower}
\end{align}
where $A'_{path} = 10^{\frac{-A_{path}}{10}}$.

Plugging (\ref{eq:noise}) and (\ref{eq:rcvpower}) into (\ref{eq:snr}), yields
\begin{align}
 \begin{split}
SNR_u^{path}(h,r_{u,n})=\frac{P_{tx}A'_{path}(d_{u,n})^{-\gamma_{path}}}{BN'},
\end{split}
 \label{eq:snr_simple}
\end{align}
where $N' = 10^{ \frac{-174+\delta_{ue}}{10}}\times10^{-3}$.

Moreover, 
the Signal to Interference plus Noise Ratio (SINR) of user $u$ connected to drone $n$ can be expressed as
\begin{align}
 \begin{split}
SINR_{u,n}^{path}(h,r_{u,n})&=\frac{S^{path}(h,r_{u,n})}{I_u+N_u}\\
&= \frac{S^{path}(h,r_{u,n})}{\big(\sum_{i \in N, i \not= n, r_{u,i} \leq \kappa } S^{path}(h,r_{u,i})\big)+N_u},
\end{split}
 \label{eq:sinr}
\end{align}
where $I_u$ (in watt) represents the interference signal from neighbor cells received by user $u$. Assuming an interference distance $\kappa$,
all neighbor drones up to distance $\kappa$ which are transmitting data to their active users create interference signals for the user $u$. It is important to note that our SINR analysis applies to both the FDMA and the TDMA cases.

Note that in this paper we focus on the analysis of small cell networks (SCNs) with an orthogonal deployment in the existing macrocell networks,
where small cells and macrocells operate on different frequency spectrum,
i.e., Small Cell Scenario \#2a defined in~\cite{TR36.872}.
Indeed,
the orthogonal deployment of dense SCNs within the existing macrocell networks has been selected as the workhorse for capacity enhancement in the 3rd Generation Partnership Project (3GPP) 4th-generation (4G) and the 5th-generation (5G) networks.
This is due to its large spectrum reuse and its easy management~\cite{Tutor_smallcell};
the latter one arising from its low interaction with the macrocell tier, e.g., no inter-tier interference.

The \textit{spectral efficiency (SE)} (bps/Hz) of an active user $u$ associated with drone $n$ can be formulated according to the Shannon Capacity Theorem as~\cite{Book_Proakis}
\begin{align}
 \begin{split}
\Phi_u^{path}(h,r_{u,n}) = \log_2 (1+SINR_{u,n}^{path}).
\end{split}
 \label{eq:individualspec}
\end{align}

Given the probabilistic channel model,
the average SE for user $u$ can be expressed as
\begin{align}
 \begin{split}
\bar{\Phi}_u(h,r_{u,n}) =&P^{LoS}(h,r_{u,n})  \Big( \log_2 (1+\frac{S^{LoS}(h,r_{u,n})}{I_u+N_u})  \Big) \\
 +& P^{NLoS}(h,r_{u,n})  \Big( \log_2 (1+\frac{S^{NLoS}(h,r_{u,n})}{I_u+N_u})  \Big).
\end{split}
 \label{eq:averagespec}
\end{align}

Then, the average SE for each cell can be computed from
\begin{align}
 \begin{split}
\bar{\Phi}(n) = \frac{\sum_{u=1}^{U} \bar{\Phi}_u(h,r_{u,n}) }{U}.
\end{split}
 \label{eq:secell}
\end{align}

Consequently, the average SE of the considered $N$-cell system can be obtained by
\begin{align}
 \begin{split}
\bar{\Phi} = \frac{\sum_{n=1}^{N} \bar{\Phi}(n) }{N}.
\end{split}
 \label{eq:sesystem}
\end{align}

Note that $\bar{\Phi}$ should be a function of time $t$ because both $r_{u,n}$ and $I_u$ may change with time due to the movement of drones.
In order not to make the notations unnecessarily complicated,
we omit $t$ in the expressions of $\bar{\Phi}$, $r_{u,n}$ and $I_u$.
We will explicitly state otherwise if $t$ should be considered in the corresponding expressions.

According to the definition of the average SE,
the average user data rate (in bits/sec) can be written as
\begin{align}
 \begin{split}
\bar{\mathpzc{R}}_u = \bar{\Phi}_u(h,r_{u,n}) \times b_u.
\end{split}
 \label{eq:datarate}
\end{align}

In the considered multi-user multi-drone system,
we define a fairness metric according to the Jain index to evaluate the fairness among the users,
which is formally presented as \cite{jain1984quantitative}
\begin{align}
 \begin{split}
\mathpzc{J}(\bar{\mathpzc{R}}_1,\bar{\mathpzc{R}}_2, \dots , \bar{\mathpzc{R}}_U) = \frac{(\sum_{u=1}^{U} \bar{\mathpzc{R}}_u)^2}{U\sum_{u=1}^{U} (\bar{\mathpzc{R}}_u)^2}.
\end{split}
 \label{eq:jainindex}
\end{align}

\subsection{The Proposed Optimization Problem}\label{sec:problem}

In this paper, 
we study the optimization problem to find the best direction for each drone at a certain time $t$ with the objective of maximizing $\bar{\Phi}$ in a small time slot $\Delta t$.

The problem can be formulated as
\begin{align}
 \begin{split}
  (\omega^*_1,\dots,\omega^*_N)^t &= \smash{\displaystyle\max_{}} \quad \bar{\Phi}(t+{\Delta t}) \\
s.t \ \ & \left|r_{u,n}(t+{\Delta t})-r_{u,n}(t)\right|\leq{v{\Delta t}}.
\end{split}
 \label{eq:opt}
\end{align}
where $\omega^*_i$ denotes the optimal direction of the $i$th drone small cell at time $t$.
Note that in Problem (\ref{eq:opt}),
$t$ and $\Delta t$ should be considered in the expressions of $\bar{\Phi}$ as well as $r_{u,n}$ and $I_u$.
Also note that the constraint of Problem (\ref{eq:opt}) is due to the finite moving distance of drones in $\Delta t$.

For clarification,
all of the parameters defined in this section are summarized in Table~\ref{tbl:params}. 

\begin{table}[!t]
    \centering
\caption{Symbol and Definition of Parameters}
\label{tbl:params}
\begin{tabular}{ll}
\hline
{\bf Parameter}              & {\bf Definition}    		 \\ \hline
$N$					& Number of Cells	\\
$U$					&Number of Users in Each Cell			\\
$B$					& Total Bandwidth							\\
$R$					& Radius of a Disk-shaped Cell \\
$l$					& Edge Length of a Square Cell		\\
$f$					&Working Frequency							\\
$h$		& Drone Height 		\\
$v$		& Drone Speed				\\
$P_{tx}$		& Drone transmission power 						\\
$\lambda$		& Mean Reading Time	 \\
$A$		&Reference Distance Path Loss (LoS/NLoS)			\\
$\gamma$		&  Path Loss Exponent (LoS/NLoS)				\\
$\delta_{ue}$               & UE Noise Figure                     \\
$\alpha, \beta$               & Environmental Parameter for Urban Area      \\
$\kappa$   &Interference Distance \\
$s$		&Data Size			\\
$\Delta t$   &Time Slot  \\

\hline
\end{tabular}
\end{table}

\section{Analysis for a Single-Drone Single-User System} \label{sec:analymodel}

To study the fundamentals on the potential gain that dynamic repositioning of drones can achieve,
we first analyze a model with just one drone BS serving a randomly deployed static user in a disk-shaped cell.
We assume a cell with a radius $R$ and one static user $u$ located at the ground distance $r_0$ ($r_0\leq R$) from the center of the cell.

\subsection{Main Results for a Hovering Drone}\label{sec:resulthovering}

First,
we assume that there is a drone that can be deployed at a low altitude $h$, 
hovering above the center of target cell and serving the user located at $r_{0}$ during the transmission time $\tau$.
Given the fact that both the user and the drone are not moving,
the spectral efficiency of the user does not change during the transmission time and can be formulated as
\begin{align}
 \begin{split}
\bar{\Phi}_{hov}(h,r_0) =&P^{LoS}(h,r_0) \Big(\log_2 (1+SNR^{LoS}(h,r_0))  \Big) \\
 +& P^{NLoS}(h,r_0)  \Big( \log_2 (1+SNR^{NLoS}(h,r_0))  \Big).
\end{split}
 \label{eq:seanalyfixed}
\end{align}

In reality,
the user can be located anywhere in the cell with various ground distances.
Without loss of generality,
we assume the user is located in the cell following a uniform distribution.
As a result,
the probability density function for the ground distance of the user from the center of the target cell can be calculated as
\begin{align}
 \begin{split}
f(x) = \frac{2\pi{x}}{\pi{R}^2}.
\end{split}
 \label{eq:pdf}
\end{align}

Therefore,
the expected SE for a user receiving data from the hovering drone BS at height $h$ can be computed by
\begin{align}
 \begin{split}
\bar{\Phi}_{hov}(h) &= \int_{r=0}^{R}\bar{\Phi}_{hov}(h,r)f(r)dr.
\end{split}
 \label{eq:expected_se_anal}
\end{align}

\subsection{Main Results for a Dynamic Repositioning Drone}\label{sec:resultmobile}

Next,
we investigate the SE performance for a dynamic repositioning drone BS, 
which is initially located above the center of target area.
Here,
when the user is randomly and uniformly deployed in the cell,
the drone will move towards the user with a low speed of $v$ while it is transmitting data to the user during the transmission time $\tau$.
During $\tau$,
the data rate increases because of the shortening distance between the drone and the user and the increasing elevation angle from the drone to the user,
which results in a larger probability of having a LoS connection.

Let us start with a static user located at the ground distance $r_0$ from the initial location of dynamic repositioning drone BS.
Depending on the drone's speed and the initial ground distance toward the user,
the drone may finish transmitting before reaching the user.
Otherwise,
the drone can reach the user before $\tau$ (i.e., $r_0$ reduces to zero),
and continues transmitting while it is hovering on top of the user.
Therefore,
we define a moving time $t_m$,
which is equal to $\min\{\frac{r_0}{v},\tau\}$ and the average SE performance can be evaluated as
\begin{align}
 \begin{split}
   \bar{\Phi}_{mob}(h,r_0) =& c\frac{ \mathlarger{\int_{r=r_0}^{r_0-vt_m} \big(P^{LoS}(h,r)\Phi(h,r)\big)  dr}}{vt_m}\\
   +& c\frac{ \mathlarger{\int_{r=r_0}^{r_0-vt_m} \big(P^{NLoS}(h,r)\Phi(h,r)\big)  dr}}{vt_m}\\
     +& (1-c)\bar{\Phi}_{hov}(h,0),
\end{split}
 \label{eq:spec_anal_one_initial}
\end{align}
where $c$ is the fraction of the transmission time that the drone is moving towards the user ($\frac{t_m}{\tau}$).
Obviously, $1-c$ is the fraction of time that the drone is hovering on top of the user.

Considering the uniform distribution of the users, the average SE over all possible user locations can be calculated as
\begin{align}
 \begin{split}
\bar{\Phi}_{mob}(h) = \int_{r=0}^{R}\bar{\Phi}_{mob}(h,r)f(r)dr.
\end{split}
 \label{eq:averagemobileanaly}
\end{align}

\subsection{Comparison of the Hovering Drone and the Dynamic Repositioning Drone}\label{sec:compareanalytic}

Now we study the potential gain that dynamic repositioning of drones can achieve by comparing the SE performance of the hovering drone with the dynamic repositioning one.
The SE of the drone BS flying at $h$ can be significantly improved compared to that of the hovering drone BS at the same height $h$.
In the following \textbf{Theorem \ref{theory:segainsingle}},
we derive an approximate upper bound of such performance improvement,
which depends on the cell radius and the drone height.

\begin{theorem} \label{theory:segainsingle}
Assuming a uniform distribution of users in a disk-shaped cell,
an approximate upper bound of the SE improvement ratio for the dynamic repositioning drone over the hovering drone is given by
\begin{align}
 \begin{split}
 \frac{\bar{\Phi}_{mob}}{\bar{\Phi}_{hov}} \approx \frac{\log_2 (\frac{P_{tx}A_{LoS}'h^{-\gamma_{LoS}}}{BN'})}{\log_2 (\frac{P_{tx}A_{NLoS}'R^{-\gamma_{NLoS}}}{BN'})}.
\end{split}
 \label{eq:ratioanalytic}
\end{align}
\end{theorem}

\begin{proof}
To obtain the upper bound of the SE improvement,
we consider the case of $\tau\to\infty$.
Note that a very large $\tau$ implies that the drone has enough time to reach on top of the user and finish the transmission at that optimal position.
The optimal position is established in the sense that the probability of having a LoS connection is equal to 1,
thus achieving the maximum SE performance.
As a result,
the maximum SE of the dynamic repositioning drone should be that of a hovering drone directly above the user,
which can be derived as
\begin{align}
 \begin{split}
 \bar{\Phi}_{mob}(h)_{\tau\to\infty} &\approx \frac{2}{R^2}\int_{r=0}^{R}\bar{\Phi}_{mob}(h,0)rdr\\
            &\approx  \frac{2}{R^2}\int_{r=0}^{R}\log_2 (1+SNR^{LoS}_u(h,0))rdr \\
            &\stackrel{SNR^{LoS}_u(h,0)\gg1}{\approx} \frac{2}{R^2}\int_{r=0}^{R}\log_2 (SNR^{LoS}_u(h,0))rdr \\
            &\approx  \frac{2}{R^2}\int_{r=0}^{R}\log_2 (\frac{P_{tx}A_{LoS}'h^{-\gamma_{LoS}}}{BN'})rdr\\
            &\approx \log_2 (\frac{P_{tx}A_{LoS}'h^{-\gamma_{LoS}}}{BN'}).
\end{split}
 \label{eq:bestmobilesesingle}
\end{align}

Similarly,
the minimum SE of the hovering drone should be achieved when there exists a NLoS connection between the drone and the user. 
Such SE is given by
\begin{align}
 \begin{split}
 \bar{\Phi}_{hov}(h) &\approx \frac{2}{R^2}\int_{r=0}^{R}\bar{\Phi}_{hov}(h,r)rdr\\
            &\approx  \frac{2}{R^2}\int_{r=0}^{R}\log_2 (1+SNR^{NLoS}_u(h,r))rdr \\
            &\stackrel{SNR^{NLoS}_u(h,r)\gg1}{\approx} \frac{2}{R^2}\int_{r=0}^{R}\log_2 (SNR^{NLoS}_u(h,r))rdr \\
            &\approx  \frac{2}{R^2}\int_{r=0}^{R}\log_2 (\frac{P_{tx}A_{NLoS}'d_{r}^{-\gamma_{NLoS}}}{BN'})rdr\\
            &\approx \log_2 (\frac{P_{tx}A_{NLoS}'}{BN'})\\
            & + \frac{2}{R^2}\int_{r=0}^{R}\log_2 \big(({h^2+r^2})^{-\gamma_{NLoS}/2}\big)rdr\\
            &\stackrel{r\gg h}{\approx} \log_2 (\frac{P_{tx}A_{NLoS}'}{BN'}) - \frac{2\gamma_{NLoS}}{R^2}\int_{r=0}^{R}\log_2(r)rdr \\
            &\approx \log_2 (\frac{P_{tx}A_{NLoS}'}{BN'}) - \gamma_{NLoS}log_2(R)\\
            &\approx \log_2 (\frac{P_{tx}A_{NLoS}'R^{-\gamma_{NLoS}}}{BN'}).
\end{split}
 \label{eq:besthoveringsesingle}
\end{align}

Then,
an approximate upper bound of the spectral efficiency ratio of the dynamic repositioning drone over the hovering drone can be obtained by the ratio of $\bar{\Phi}_{mob}(h)_{\tau\to\infty}$ over $\bar{\Phi}_{hov}(h)$ as shown in (\ref{eq:ratioanalytic}),
which concludes our proof.
\end{proof}

In order to get some sense about how large is the SE improvement ratio derived in \textbf{Theorem \ref{theory:segainsingle}},
we provide some numerical results in Figure \ref{fig:analytical} with the parameters set to $R=40m$, $v={10,15,20}m/s$, $s=2MByte$, $P_{tx} = 24dBm$, and $B=10MHz$.
\begin{figure}[!t]
\centering
  \includegraphics[scale=0.44]{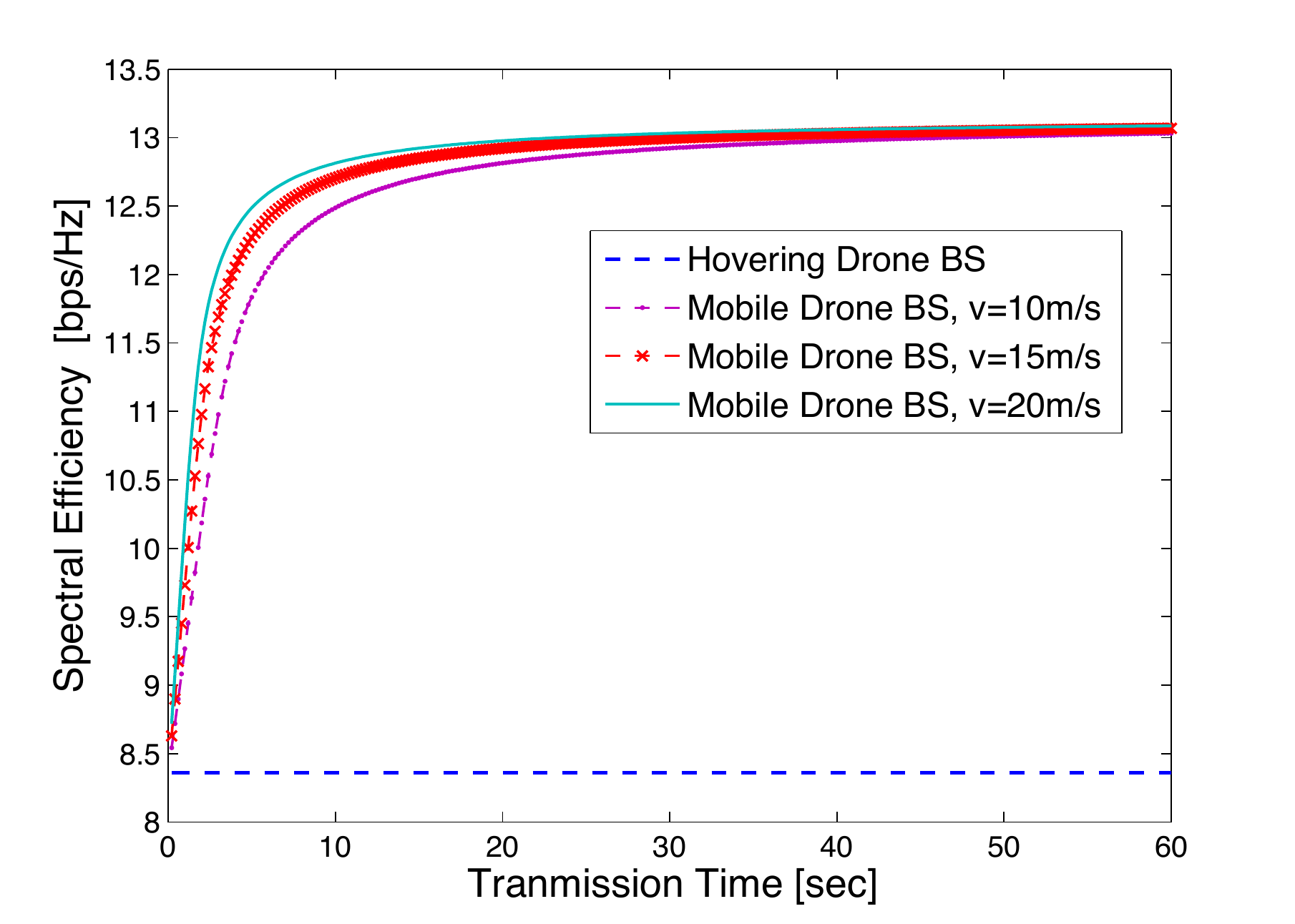}
\caption{SE for hovering and dynamic repositioning drone at height 10m vs. transmission time for different drone speed}
\label{fig:analytical}
\vspace{-0.4cm}
\end{figure}
From Figure \ref{fig:analytical},
we can observe that the dynamic repositioning drone flying at $h=10m$ can improve the SE up to 57\% compared with the hovering drone BS.
The SE is getting larger for the dynamic repositioning drone with a longer transmission time $\tau$,
because a larger $\tau$ implies that the drone has more time to fly closer to the user.

Moreover,
a higher drone speed promises a higher SE.
The reason is that a drone with a higher speed can fly closer to the user in a shorter time frame than that with a lower speed.
However,
the benefit of the high drone speed wanes as the transmission time increases.
Hereafter,
we focus on the drone speed of $v=10m/s$,
which is a practical speed for the drone as discussed in Section~\ref{sec:intro}.
Moreover,
according to~\cite{7101619} and as discussed in Section~\ref{sec:intro},
the power consumption of a drone moving with a speed of $10m/s$ is almost equal to the power consumption of a hovering drone at the same height.

This huge SE improvement motivates us to optimize the moving direction for the drone BSs in order to achieve their potential SE performance.
To consider more practical networks,
multiple mobile users and multiple dynamic repositioning drones are considered in the following sections.
In the sequel,
we assume that both the hovering drone BSs and the dynamic repositioning ones are operating at the same low height,
and the benefits of the dynamic repositioning drone BSs will be quantified.

\section{The Proposed Dynamic Repositioning Scheme for a Multi-Drone Multi-User System}\label{sec:proposedmethod}

Given the huge SE improvement resulting from the dynamic repositioning of drone BSs,
we extend our scheme to the application on a multi-drone multi-mobile-user network scenario.

It is apparent that the movement of drones will dynamically change the interference pattern that affects the SE of the system.
As a result,
optimizing the moving path of drones during the operation time is one of the challenges in designing the system.
Assuming a global knowledge of all drones and all users in the system,
a centralized controller can find the optimal locations for all drones that maximize the SE (see section \ref{sec:problem}).
However, 
finding the optimal location is of high complexity and cannot be achieved in real time.
Alternatively,
at each time slot,
each drone in a cell can dynamically choose the moving direction between $0$ and $2\pi$ to maximize the SE.
In order to reduce the complexity of the problem,
discrete angles are introduced to reduce the search space.
For instance,
suppose that the angle step is $\Delta g = \pi/M$,
then $2M$ candidate directions, i.e.,
$\{0,\frac{\pi}{M}, \frac{2\pi}{M} , \dots , 2\pi - \frac{\pi}{M}\}$,
will be examined by each drone.
Then,
each drone will update its movement angle and transmit data to users at each defined time slot $\Delta t$ while moving.

Given $N$ drones in the network area,
there are $(2M)^{N}$ possible cases at each time slot, 
making the complexity of centralized controller method in the order of $O({2M}^{N})$,
which is difficult to implement in a real-time manner.

Therefore,
we propose distributed algorithms that each drone can choose a moving direction at each time slot,
independent of the other drones.
As a result,
the complexity of the system reduces to $O(2NM)$ at each time slot.

To this end,
we define two criteria and let each drone calculate the metric based on either criterion for each possible direction.
Then the drone decides its moving direction based on maximizing the criterion.

\begin{figure}[!t]
\centering
  \includegraphics[scale=0.38]{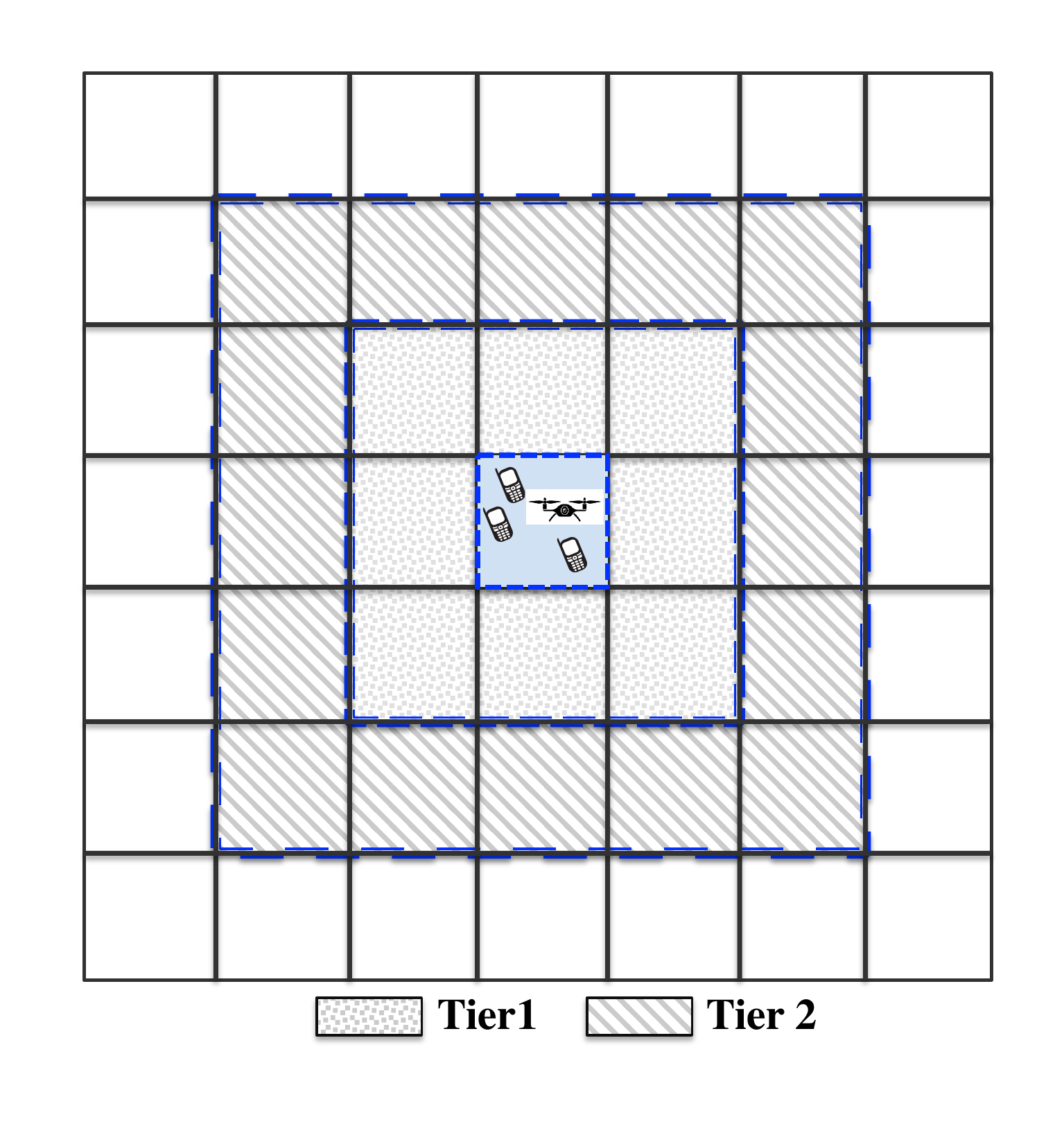}
\caption{Multi-drone multi-user system containing 49 square cells}
\label{fig:arch}
\vspace{-0.4cm}
\end{figure}

\subsection{Maximization of SNR} \label{sec:critSNR}

The first criterion is based on the maximization of \textit{SNR} (Signal to Noise Ratio) for active users in each cell.
In this strategy,
each drone focuses only on its active users during a time slot.
Each drone assumes that there are no other drones in the neighborhood,
and given the location of its own active users,
calculates the SNRs for every active user and every candidate angle.
Finding the best moving angle based on such SNR maximization criterion can be formally presented as
\begin{align}
 \begin{split}
 \omega^t_n =\  & \text{arg}\,\max\  {SNR}_n \quad  n \in \{1,\dots, N\}\\
s.t \ \ & \left|r_{u,n}(t+{\Delta t})-r_{u,n}(t)\right|\leq{v{\Delta t}}.
\end{split}
 \label{eq:snr_alg}
\end{align}

When there is no pending request, drones stay hovering at their current location.

\subsection{Maximization of SLR} \label{sec:critSLR}

The movement of each drone also changes its interference toward the active users in the neighbor cells.
As a result,
we propose another criterion for each drone that not only considers the received signal for its own active users,
but also attempts to reduce the interference to the other active users in the neighbor cells.
In this criterion,
we assume that each drone knows the location of the active users in the neighbor cells as well.
The amount of interference from drone $n$ to the other active users is referred to as \textit{Leakage}~\cite{6571248},
which is defined as follows,
\begin{align}
 \begin{split}
L^t_n=\sum_{j \in N, j\not= n, d_{jn} \leq \kappa}\big(\sum_{\forall u \in \mathpzc{Q}^t_{j}} {S^{path}}(h,r_{u,n})\big).
\end{split}
 \label{eq:leakage}
\end{align}

Based on (\ref{eq:leakage}),
each drone calculates the \textit{SLR} (Signal to Leakage Ratio) value for every candidate direction for the active users in the neighbor cells. 
Such SLR value for each user $u$ of drone $n$ can be written as
\begin{align}
 \begin{split}
SLR^t_u= \frac{S^{path}(h,r_{u,n})}{L^t_n}.
\end{split}
 \label{eq:slr}
\end{align}

Similar to SNR criteria, 
when there is no pending request, 
drones stay hovering at their current location. 
Finding the best moving angle based on such SLR maximization criterion can be formally presented as
\begin{align}
 \begin{split}
 \omega^t_n =& \text{arg}\,\max\  {SLR}_n  \quad  n \in \{1,\dots, N\}\\
s.t \ \ & \left|r_{u,n}(t+{\Delta t})-r_{u,n}(t)\right|\leq{v{\Delta t}}.
\end{split}
 \label{eq:slr_alg}
\end{align}

\section{Simulation Results} \label{sec:simulationresults}

\begin{figure}[!t]
\centering
  \includegraphics[scale=0.42]{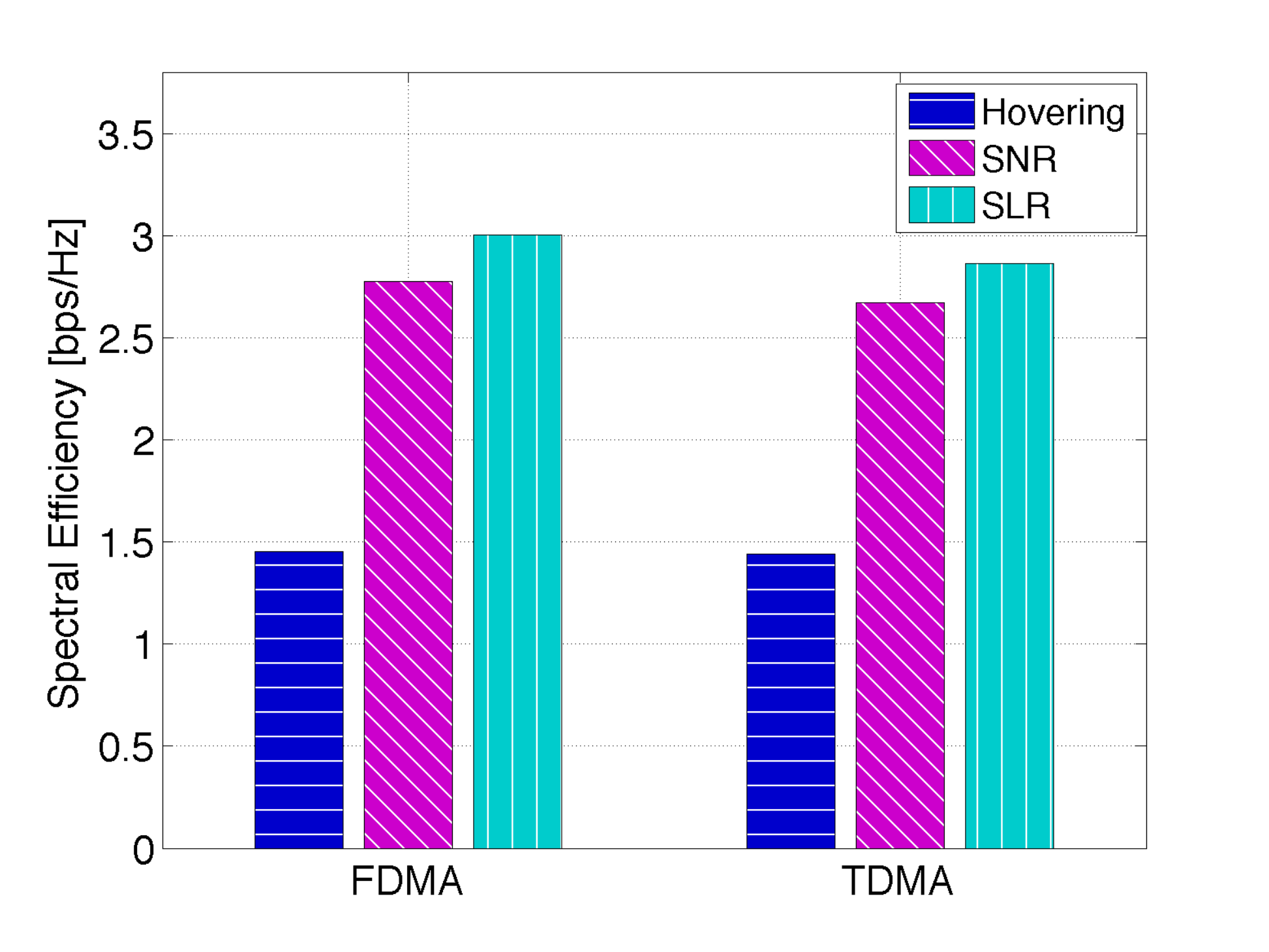}
\caption{Spectral efficiency for both hovering drone BSs and dynamic repositioning drone BSs at low altitude $h=$10m}
\label{fig:se}
\vspace{-0.3cm}
\end{figure}

In this section,
the performance of our proposed algorithms is evaluated in a multi-drone multi-mobile-user small cell network through simulation using MATLAB.

\begin{figure*}
    \centering
    \begin{subfigure}[b]{0.32\textwidth}
        \includegraphics[width=1.1\textwidth,height=1\textwidth]{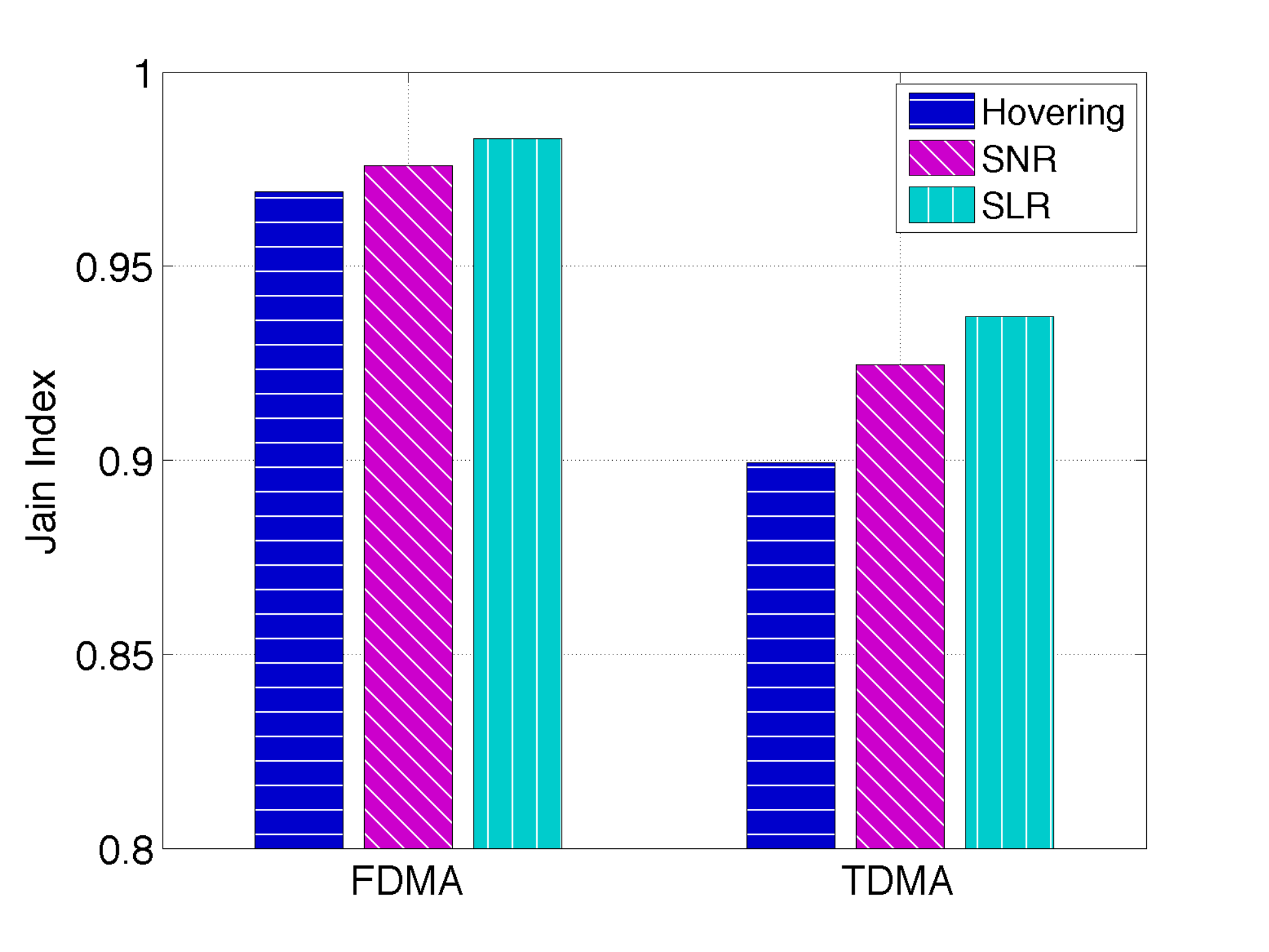}
        \caption{Jain Index for dynamic repositioning and hovering drone BSs}
        \label{fig:jain}
    \end{subfigure}
    ~ 
    \begin{subfigure}[b]{0.32\textwidth}
        \includegraphics[width=1.1\textwidth,height=1\textwidth]{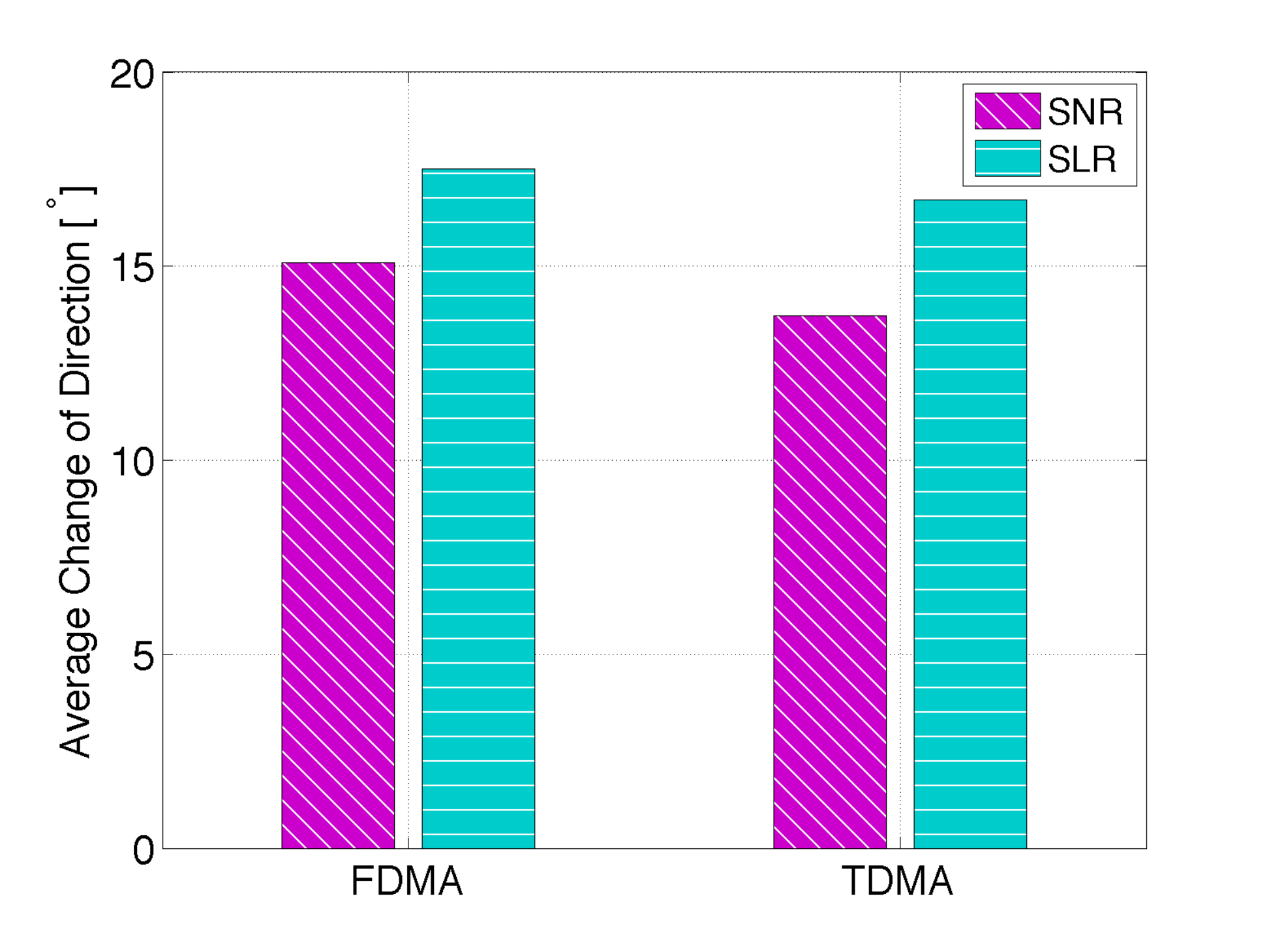}
        \caption{The average changing direction value for dynamic repositioning drone BSs}
        \label{fig:angle}
    \end{subfigure}
    ~ 
    \begin{subfigure}[b]{0.32\textwidth}
        \includegraphics[width=1.1\textwidth,height=1\textwidth]{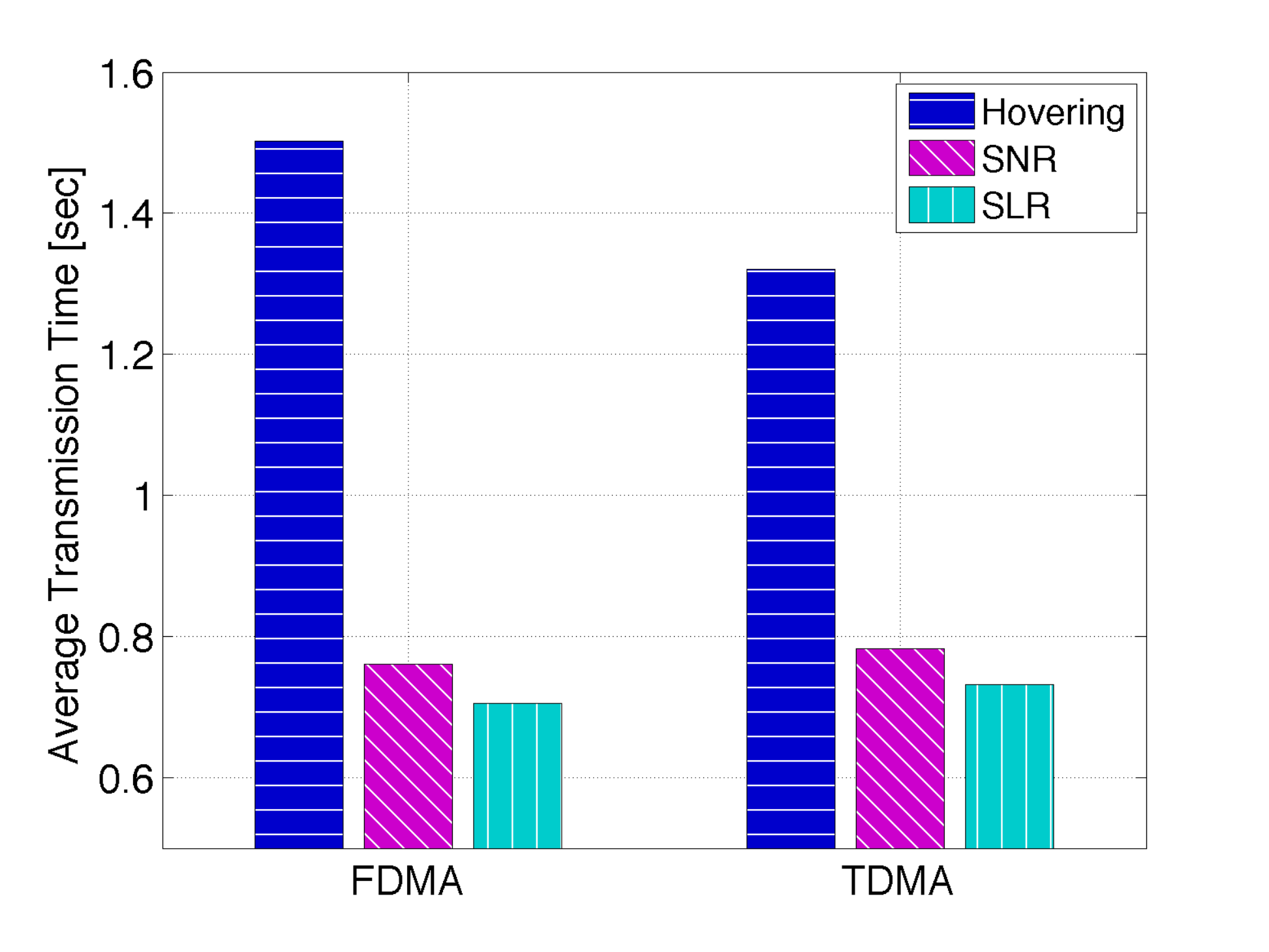}
        \caption{The average transmission time for each user request}
        \label{fig:txtime}
    \end{subfigure}
    \caption{Comparison of (a) Jain index, (b) average turning angles, (c) average transmission time for hovering and dynamic repositioning drones }\label{fig:simulmetrics}
\end{figure*}

\subsection{Simulation Scenarios and Parameters} \label{sec:sim_scen}
Here,
we consider 49 square cells,
each with an edge length of 80m,
which can be approximated by a disk with a radius of 40m.
In each cell,
there are 5 users moving according to a random way point (RWP) model. 
In the RWP model, 
each user selects a random destination independent of other users and moves there with a random speed. 
As a result, 
user would be uniformly distributed in the cell during the simulation time.
Following the recommended 3GPP traffic model,
the mean reading time for each user is set to 20sec,
and each request of user generates a 2MByte data packet\cite{3gpp36814}.
Unless otherwise stated,
parameter values used in simulations are set according to Table~\ref{tbl:values}.

We consider an average interference distance of 200m,
and hence,
each drone creates interference for up to two tiers of neighbor cells as illustrated in Figure \ref{fig:arch}.
From Figure~\ref{fig:arch},
we can see that outer cells will receive lower interference than inner cells.
To obtain unbiased results,
we collect the data just from the 9 inner cells as shown in Figure \ref{fig:arch}.
Moreover,
to mitigate the randomness of the results,
all results have been averaged over 10 independent runs of 400-second simulations.

\begin{table}[]
\centering
\caption{Parameter values in our simulation}
\label{tbl:values}
\begin{tabular}{lll}
\hline
{\bf Symbol}              & {\bf Definition}    	 & {\bf Value}  	 \\
\hline
\hline
$N$					& Number of Drones	& 49\\
$h$		& Drone Height 		                                 &10 m\\
$v$		& Drone Speed			                                   &10 m/s	\\
$B$					& Total Bandwidth						&10 MHz	\\
$f$					&Working Frequency					&2 GHz		\\
$P_{tx}$		& Drone Transmission Power 	        & 24 dBm	\cite{TR36.828}\\
$\delta_{ue}$               & UE Noise Figure                  & 9 dB      \\
$U$					&Number of Users in Each Cell	 &5		\\
$l$					& Edge Length of a Square Cell	&80 m	\\
$\lambda$		& Mean Reading Time	 & 20 sec\\
$A$		&Reference Distance Path Loss (LoS/NLoS)	&	41.1/33 \cite{TR36.828}	\\
$\gamma$		&  Path Loss Exponent (LoS/NLoS)	&	 2.09/3.75 \cite{TR36.828} 		\\
$\alpha, \beta$               & Environmental Parameter for Urban Area   & 9.61, 0.16 \cite{yaliniz2016efficient}\\
$\kappa$   &Interference Distance  & 200 m\\
$s$		&Data Size		& 2 MByte	\\
$\Delta t$   &Time Slot  & 100 msec\\
$T$		&Simulation Time	 & 400 sec	\\
$\Delta g$   &Angle Step  & $5^{\circ}$ \\
\hline
\end{tabular}
\end{table}


\subsection{The SE Performance} \label{sec:sim_SE}

In Figure \ref{fig:se},
we plot the SE performance for both the hovering drone BSs and the dynamic repositioning drone BSs at a low altitude of $h=10m$.
From this figure, we can draw the following observations:
\begin{itemize}
  \item The SE performance suffers from a huge loss in the face of inter-cell interference from neighbor cells compared to the single-cell system.
  \item Moreover,
it can be concluded that the introduction of the dynamic repositioning capability to drones can significantly improve the SE performance without any negative effect on the energy consumption.
More specifically,
as mentioned earlier,
the power consumption of a drone moving at 10m/s is almost equal to the power consumption of a hovering drone at the same height \cite{7101619}.
  \item The SLR criterion performs slightly better (about 7-8\%) than the SNR criterion,
  which can be attributed to the fact that the SLR criterion has more information than the SNR one,
  and can manage the drone movement while reducing interference to neighboring cell users.
  \item The performance gain is equally achievable by means of FDMA and TDMA,
  which implies a wide application for the proposed dynamic repositioning scheme.
\end{itemize}

\subsection{The fairness Performance} \label{sec:sim_fairness}

Furthermore,
we study the fairness performance for the proposed criteria and the proposed bandwidth allocation algorithms.
In Figure \ref{fig:jain},
we show the fairness in terms of the Jain's index for the interested schemes.
From this figure, we can draw the following observations:
\begin{itemize}
  \item The dynamic repositioning drone BSs can deliver better fairness compared to the hovering ones.
The reason is that the dynamic repositioning drones can change their location in order to improve the overall SE,
thus eliminating the notion of cell-edge users and resulting in a higher fairness metric.
  \item The fairness metric for the FDMA method is higher than that of the TDMA method.
The reason is that the TDMA method allocates the whole bandwidth to just one user at each time slot,
and hence other users may experience bandwidth starvation,
resulting in a lower fairness metric.
\end{itemize}

\subsection{Drone Turning Angles}\label{sec:sim_angle}

As dynamic repositioning requires the drone to change its direction every $\Delta t$ sec, 
we analyze the average degree of change of direction in Figure~\ref{fig:angle}. 
It is encouraging to see that the average turning angle is very small and confined to between 14 and 18 degree depending on the choice of our control algorithms. 

\subsection{Tranmission Time Performance} \label{sec:txttime}

Following the improvement of spectral efficiency of dynamic repositioning drones observed in the Figure \ref{fig:se},
the average transmission time for each user request is reduced by our proposed methods as well.
As illustrated in Figure \ref{fig:txtime},
dynamic repositioning drones can decrease the transmission time by around 50\%.

%

\section{Conclusion and Future Work}\label{sec:conclusion}

We have proposed dynamic repositioning of drone BSs as a novel method to increase the spectral efficiency of drone small cells at a low altitude, e.g., 10m. 
We proposed two distributed algorithms that can be used to autonomously reposition the drones in response to user activities and mobility. 
We have shown that dynamic repositioning of drones can nearly double the spectral efficiency in a multi-drone multi-mobile-user small cell network, 
while its energy consumption can be kept at the same level as the network where drones are hovering above pre-determined positions.

Our work opens up several directions of new research. 
With advancements in camera technology and video processing, 
drones in the future could identify obstacles in real time. 
Such knowledge could be used to design more advanced dynamic repositioning algorithms to largely increase the probability of LoS communications with active users. 
Another direction would be to actually implement the proposed algorithms in real drones and conduct real-life experiments. 
Such experiments will provide insights to more practical issues and help improve the dynamic repositioning algorithms.

\bibliographystyle{IEEEtran}
\bibliography{mainwowmom}

\begin{thebibliography}{10}
\providecommand{\url}[1]{#1}
\csname url@samestyle\endcsname
\providecommand{\newblock}{\relax}
\providecommand{\bibinfo}[2]{#2}
\providecommand{\BIBentrySTDinterwordspacing}{\spaceskip=0pt\relax}
\providecommand{\BIBentryALTinterwordstretchfactor}{4}
\providecommand{\BIBentryALTinterwordspacing}{\spaceskip=\fontdimen2\font plus
\BIBentryALTinterwordstretchfactor\fontdimen3\font minus
  \fontdimen4\font\relax}
\providecommand{\BIBforeignlanguage}[2]{{%
\expandafter\ifx\csname l@#1\endcsname\relax
\typeout{** WARNING: IEEEtran.bst: No hyphenation pattern has been}%
\typeout{** loaded for the language `#1'. Using the pattern for}%
\typeout{** the default language instead.}%
\else
\language=\csname l@#1\endcsname
\fi
#2}}
\providecommand{\BIBdecl}{\relax}
\BIBdecl

\bibitem{Namuduri}
K.~Namuduri, Y.~Wan, and M.~Gomathisankaran, ``Mobile ad hoc networks in the
  sky: State of the art, opportunities, and challenges,'' in \emph{Proceedings
  of the Second ACM MobiHoc Workshop on Airborne Networks and Communications},
  ser. ANC '13.\hskip 1em plus 0.5em minus 0.4em\relax New York, NY, USA: ACM,
  2013, pp. 25--28.

\bibitem{6399121}
S.~Rohde and C.~Wietfeld, ``Interference aware positioning of aerial relays for
  cell overload and outage compensation,'' in \emph{Vehicular Technology
  Conference (VTC Fall), 2012 IEEE}, Sept 2012, pp. 1--5.

\bibitem{7122576}
A.~Merwaday and I.~Guvenc, ``{UAV} assisted heterogeneous networks for public
  safety communications,'' in \emph{Wireless Communications and Networking
  Conference Workshops (WCNCW), 2015 IEEE}, March 2015, pp. 329--334.

\bibitem{7417609}
M.~Mozaffari, W.~Saad, M.~Bennis, and M.~Debbah, ``Drone small cells in the
  clouds: Design, deployment and performance analysis,'' in \emph{2015 IEEE
  Global Communications Conference (GLOBECOM)}, Dec 2015, pp. 1--6.

\bibitem{7451189}
V.~Sharma, M.~Bennis, and R.~Kumar, ``Uav-assisted heterogeneous networks for
  capacity enhancement,'' \emph{IEEE Communications Letters}, vol.~20, no.~6,
  pp. 1207--1210, June 2016.

\bibitem{7122575}
X.~Li, D.~Guo, H.~Yin, and G.~Wei, ``Drone-assisted public safety wireless
  broadband network,'' in \emph{Wireless Communications and Networking
  Conference Workshops (WCNCW), 2015 IEEE}, March 2015, pp. 323--328.

\bibitem{6848000}
A.~E. A.~A. Abdulla, Z.~M. Fadlullah, H.~Nishiyama, N.~Kato, F.~Ono, and
  R.~Miura, ``An optimal data collection technique for improved utility in
  uas-aided networks,'' in \emph{IEEE INFOCOM 2014 - IEEE Conference on
  Computer Communications}, April 2014, pp. 736--744.

\bibitem{7192644}
K.~Li, W.~Ni, X.~Wang, R.~P. Liu, S.~S. Kanhere, and S.~Jha, ``Energy-efficient
  cooperative relaying for unmanned aerial vehicles,'' \emph{IEEE Transactions
  on Mobile Computing}, vol.~15, no.~6, pp. 1377--1386, June 2016.

\bibitem{7090210}
G.~Stamatescu, D.~Popescu, and R.~Dobrescu, ``Cognitive radio as solution for
  ground-aerial surveillance through wsn and uav infrastructure,'' in
  \emph{Electronics, Computers and Artificial Intelligence (ECAI), 2014 6th
  International Conference on}, Oct 2014, pp. 51--56.

\bibitem{al2014optimal}
A.~Al-Hourani, S.~Kandeepan, and S.~Lardner, ``Optimal {LAP} altitude for
  maximum coverage,'' \emph{Wireless Communications Letters, IEEE}, vol.~3,
  no.~6, pp. 569--572, Dec 2014.

\bibitem{7510870}
M.~Mozaffari, W.~Saad, M.~Bennis, and M.~Debbah, ``Optimal transport theory for
  power-efficient deployment of unmanned aerial vehicles,'' in \emph{2016 IEEE
  International Conference on Communications (ICC)}, May 2016, pp. 1--6.

\bibitem{yaliniz2016}
R.~I. Bor-Yaliniz, A.~El-Keyi, and H.~Yanikomeroglu, ``Efficient 3-d placement
  of an aerial base station in next generation cellular networks,'' in
  \emph{2016 IEEE International Conference on Communications (ICC)}, May 2016,
  pp. 1--5.

\bibitem{7037248}
A.~Al-Hourani, S.~Kandeepan, and A.~Jamalipour, ``Modeling air-to-ground path
  loss for low altitude platforms in urban environments,'' in \emph{2014 IEEE
  Global Communications Conference}, Dec 2014, pp. 2898--2904.

\bibitem{7848883}
A.~Fotouhi, M.~Ding, and M.~Hassan, ``Dynamic base station repositioning to
  improve performance of drone small cells,'' in \emph{2016 IEEE Globecom
  Workshops (GC Wkshps)}, Dec 2016, pp. 1--6.

\bibitem{Ding2016GC_ASECrash}
M.~{Ding} and D.~{Lopez-Perez}, ``Please lower small cell antenna heights in
  {5G},'' \emph{IEEE Globecom 2016, [Online] {arXiv:1611.01869 [cs.IT]}}, pp.
  1--6, Dec. 2016.

\bibitem{ZORBAS201380}
E.~M. Shakshuki, D.~Zorbas, T.~Razafindralambo, D.~P.~P. Luigi, and
  F.~Guerriero, ``Energy efficient mobile target tracking using flying
  drones,'' \emph{Procedia Computer Science}, vol.~19, pp. 80 -- 87, 2013.

\bibitem{Rubio2016}
\BIBentryALTinterwordspacing
A.~A. Rubio, A.~Seuret, Y.~Ariba, and A.~Mannisi, \emph{Optimal Control
  Strategies for Load Carrying Drones}.\hskip 1em plus 0.5em minus 0.4em\relax
  Cham: Springer International Publishing, 2016, pp. 183--197. [Online].
  Available: \url{http://dx.doi.org/10.1007/978-3-319-32372-5_11}
\BIBentrySTDinterwordspacing

\bibitem{7101619}
C.~D. Franco and G.~Buttazzo, ``Energy-aware coverage path planning of uavs,''
  in \emph{Autonomous Robot Systems and Competitions (ICARSC), 2015 IEEE
  International Conference on}, April 2015, pp. 111--117.

\bibitem{TR36.828}
3gpp, ``3gpp tr 36.828,further enhancements to lte time division duplex for
  downlink-uplink interference management and traffic adaptation,'' 3GPP
  Technical Report, Tech. Rep., 2012.

\bibitem{3gpp36814}
3GPP, ``3gpp tr 36.814 version 9.0.0, release 9: 3rd generation partnership
  project; further advancements for e-utra physical layer aspects,'' 3GPP
  Technical Report, Tech. Rep., 2010.

\bibitem{thermalnoise}
C.~S. Turner, ``Johnson-nyquist noise,'' \emph{url: http://www. claysturner.
  com/dsp/Johnson-NyquistNoise}, 2012.

\bibitem{TR36.872}
3GPP, ``{TR 36.872: Small cell enhancements for E-UTRA and E-UTRAN - Physical
  layer aspects},'' Dec. 2013.

\bibitem{Tutor_smallcell}
D.~Lopez-Perez, M.~Ding, H.~Claussen, and A.~Jafari, ``{Towards 1 Gbps/UE in
  cellular systems: Understanding ultra-dense small cell deployments},''
  \emph{IEEE Communications Surveys Tutorials}, vol.~17, no.~4, pp. 2078--2101,
  Jun. 2015.

\bibitem{Book_Proakis}
J.~G. Proakis, \emph{{Digital Communications (4th Ed.)}}.\hskip 1em plus 0.5em
  minus 0.4em\relax {New York: McGraw-Hill}, 2000.

\bibitem{jain1984quantitative}
R.~Jain, D.-M. Chiu, and W.~R. Hawe, \emph{A quantitative measure of fairness
  and discrimination for resource allocation in shared computer system}.\hskip
  1em plus 0.5em minus 0.4em\relax Eastern Research Laboratory, Digital
  Equipment Corporation Hudson, MA, 1984, vol.~38.

\bibitem{6571248}
M.~Ding, M.~Zhang, H.~Luo, and W.~Chen, ``Leakage-based robust beamforming for
  multi-antenna broadcast system with per-antenna power constraints and
  quantized cdi,'' \emph{IEEE Transactions on Signal Processing}, vol.~61,
  no.~21, pp. 5181--5192, Nov 2013.

\bibitem{yaliniz2016efficient}
R.~I. Bor-Yaliniz, A.~El-Keyi, and H.~Yanikomeroglu, ``Efficient 3-d placement
  of an aerial base station in next generation cellular networks,'' in
  \emph{2016 IEEE International Conference on Communications (ICC)}, May 2016,
  pp. 1--5.

\end{thebibliography}

\end{document}